\documentclass[aps,pra,reprint,superscriptaddress,amsmath,amssymb,nofootinbib]{revtex4-2}

\usepackage{graphicx}
\usepackage{dcolumn}
\usepackage{bm}
\usepackage[dvipsnames]{xcolor}
\usepackage{pgfplots}
\usepackage{braket}
\usepackage{verbatim}

\definecolor{darkblue}{rgb}{0,0.3,0.7}
\usepackage[pdftex,colorlinks=true,bookmarks=false,citecolor=blue,urlcolor=blue]{hyperref}
\hypersetup{
colorlinks=true,
linkcolor=darkblue,
filecolor=blue,
citecolor=darkblue,  
urlcolor=darkblue,
}

\DeclareMathOperator{\Tr}{Tr}

\newcommand\imag{\mathrm{i}}

\usepackage[capitalize]{cleveref}
\crefname{section}{Sec.}{Secs.}
\Crefname{section}{Section}{Sections}

\crefrangelabelformat{section}{#3#1#4--#5\crefstripprefix{#1}{#2}#6}
\crefrangelabelformat{figure}{#3#1#4--#5\crefstripprefix{#1}{#2}#6}
\crefrangelabelformat{equation}{(#3#1#4--#5\crefstripprefix{#1}{#2}#6)}

\crefmultiformat{equation}%
{\edef\crefstripprefixinfo{#1}Eqs.~(#2#1#3}%
{,#2\crefstripprefix{\crefstripprefixinfo}{#1}#3)}%
{,#2\crefstripprefix{\crefstripprefixinfo}{#1}#3}%
{,#2\crefstripprefix{\crefstripprefixinfo}{#1}#3)}
\pgfplotsset{compat=1.18}

\usepackage{amsthm}
\newtheorem{theorem}{Theorem}
\usepackage{tikz}
\usetikzlibrary{shapes, arrows.meta, positioning, calc, decorations.pathmorphing}

\begin{document}

\title{Induced random mixed states are symmetric Dirichlet mixtures}

\author{Brian T. Kirby}
\affiliation{DEVCOM Army Research Laboratory, Adelphi, MD 20783 USA}
\affiliation{Tulane University, New Orleans, LA 70118 USA}
\email{Brian.T.Kirby4.civ@army.mil}

\author{Alexander C. B. Greenwood}
\affiliation{
Dept of Electrical \& Computer Engineering, University of Toronto, Toronto, Ontario, Canada M5S 3G4
}

\author{Sanjaya Lohani}
\affiliation{Department of Electrical and Computer Engineering, Southern Methodist University, Dallas, Texas 75205, USA}

\author{Joseph M. Lukens}
\affiliation{Elmore Family School of Electrical and Computer Engineering and Purdue Quantum Science and Engineering Institute, Purdue University, West Lafayette, Indiana 47907, USA}
\affiliation{Quantum Information Science Section, Oak Ridge National Laboratory, Oak Ridge, Tennessee 37831, USA}

\date{\today}

\begin{abstract}
We establish an exact equivalence in distribution between $D$-dimensional random mixed states induced by partial traces over $K$-dimensional environments and the Mai-Alquier distribution, a mixture of $K$ independent Haar-random pure states weighted by a symmetric Dirichlet distribution. This identification recasts a class of expectation values for induced ensembles into calculations involving Dirichlet moments and low-order Haar averages on a single-system state space. As applications, we recover exact purity moments up to fourth order, derive the mean Hilbert--Schmidt distance between independent induced ensembles with possibly different environment dimensions, and obtain exact average determinants. These results provide a unified and constructive perspective on induced random mixed states and on the evaluation of quantities such as purity moments, overlaps, and determinants.
\end{abstract}

\maketitle

\section{Introduction}

The characterization of random mixed quantum states is a foundational problem with broad applications across open quantum systems, quantum statistical mechanics, and quantum information processing
\cite{bengtsson2007geometry, breuer2002theory, popescu2006entanglement, collins2016random}. 
Traditionally, this topic has been approached through the physical paradigm of environmental coupling: when a $D$-dimensional system and a $K$-dimensional environment are prepared in a uniformly Haar-random global pure state, tracing out the environment yields an induced measure on the system of interest \cite{zyczkowski2001induced}. Analyzing these reduced density matrices has historically relied on random matrix theory (RMT), employing continuous integration over unitary manifolds via Weingarten calculus \cite{collins2006integration, collins2010random} or joint eigenvalue distributions \cite{mehta2004random, sommers2004statistical}. While useful, these mathematical techniques are often computationally intensive.

Parallel to these physically motivated constructions, the quantum state tomography and Bayesian inference communities have explored a complementary generative model for random mixed states: the Mai-Alquier (MA) ensemble \cite{mai2017pseudo, Lukens2020b, lohani2021improving, Lohani2022a, Mai2023, lohani2023demonstration}. In this construction, a mixed state is generated as a convex combination of $K$ independent Haar-random pure states with weights drawn from a symmetric Dirichlet distribution. This ensemble has proven useful as a tunable prior for simulation and inference, particularly because its purity can be adjusted through the Dirichlet concentration parameter. Although the relation between Dirichlet weight normalization and Wishart-type constructions is suggestive, the induced-state and MA constructions are typically presented in different languages and used for different purposes. 

These constructions admit natural interpretations in terms of the distinction between ``proper" and ``improper" mixtures discussed by d'Espagnat \cite{d2018conceptual}. For each sampled MA decomposition, preparing its component pure states with the specified probabilities provides a proper-mixture realization. By contrast, tracing out part of an entangled global pure state yields an improper mixture, even when the global state is completely known. This distinction concerns the preparation and description of a state; identical density operators nevertheless yield identical statistics for measurements restricted to the system.

In this work, we make this connection explicit by proving that the induced ensemble with environment dimension $K$ is exactly equal in distribution to the symmetric MA ensemble with $K$ pure-state components and Dirichlet concentration parameter $\alpha=D$. Our contribution is therefore not to replace the standard induced-state formalism, but to provide an equivalent single-system representation in which certain observables can be evaluated through Dirichlet moments and elementary Haar averages rather than through direct integration over composite manifolds or joint eigenvalue densities. We illustrate the usefulness of this perspective by recovering known purity moments, deriving the mean Hilbert-Schmidt distance between independent ensembles with $K_1\neq K_2$, and obtaining an exact expression for the mean determinant.

The remainder of this work is organized as follows: in Sec.~\ref{sec:equivalence}, we formally prove the distributional equivalence between induced measures and MA ensembles; in Sec.~\ref{sec:analytical}, we derive exact second and third purity moments (extending to the fourth moment in Appendix~\ref{app:exact_4th_moment}), evaluate exact Hilbert-Schmidt distances between distinct induced ensembles, and compute average matrix determinants, validating our results against known literature benchmarks; and in \cref{sec:conclusion}, we summarize our findings and propose directions for future work.
Finally, Appendix~\ref{app:generalized_dirichlet} extends these mathematical identities to generalized asymmetric Dirichlet priors.

\section{Equivalence of induced measures and Dirichlet ensembles}
\label{sec:equivalence}

\begin{figure*}[t]
\centering
\begin{tikzpicture}[
    pure state/.style={circle, fill=blue!70, draw=black, thick, inner sep=2pt},
    mixed state/.style={circle, fill=red!80, draw=black, thick, inner sep=2.5pt},
    unnorm vector/.style={-{Stealth[length=2.2mm]}, thick, orange!80!black},
    unnorm point/.style={circle, fill=orange!50, draw=orange!90!black, thick, inner sep=1.8pt},
    bloch outline/.style={draw=black!70, thick},
    bloch faint/.style={draw=black!35, dashed, thick},
    equator front/.style={draw=black!70, thick},
    equator faint/.style={draw=black!25, dashed},
    axis line/.style={draw=black!30, ->},
    hull line/.style={draw=blue!50, thick, dashed},
    font=\sffamily\small
]
\def\R{1.5} 
\def\Aone{55}
\def\Atwo{205}
\def\Athree{-35}
\def\Yone{1.45} 
\def\Ytwo{0.55}
\def\Ythree{0.80}
\pgfmathsetmacro{\sumY}{\Yone + \Ytwo + \Ythree}
\pgfmathsetmacro{\Pone}{\Yone / \sumY}
\pgfmathsetmacro{\Ptwo}{\Ytwo / \sumY}
\pgfmathsetmacro{\Pthree}{\Ythree / \sumY}
\pgfmathsetmacro{\Rhox}{\R * (\Pone*cos(\Aone) + \Ptwo*cos(\Atwo) + \Pthree*cos(\Athree))}
\pgfmathsetmacro{\Rhoy}{\R * (\Pone*sin(\Aone) + \Ptwo*sin(\Atwo) + \Pthree*sin(\Athree))}

\begin{scope}[xshift=0cm]
    
    \draw[bloch outline] (0,0) circle (\R);
    \draw[equator faint] (\R,0) arc (0:180:\R cm and 0.45cm);
    \draw[equator front] (\R,0) arc (0:-180:\R cm and 0.45cm);
    
    \draw[axis line] (0,-2.1) node[below] {$|1\rangle$} -- (0,2.1) node[above] {$|0\rangle$};
    
    \coordinate (S1) at (\Aone:\R);
    \coordinate (S2) at (\Atwo:\R);
    \coordinate (S3) at (\Athree:\R);
    \coordinate (Rho) at (\Rhox, \Rhoy);
    
    \draw[hull line] (S1) -- (S2) -- (S3) -- cycle;
    
    \node[pure state, label={[blue!80!black]above right:$|\psi_1\rangle$}] at (S1) {};
    \node[pure state, label={[blue!80!black]left:$|\psi_2\rangle$}] at (S2) {};
    \node[pure state, label={[blue!80!black]right:$|\psi_3\rangle$}] at (S3) {};
    \node[mixed state, label={[red!80!black]below:$\rho_{\text{MA}}^{(K)}$}] at (Rho) {};
    
    \draw[->, red!50, shorten >=2pt, shorten <=2pt] (S1) -- (Rho);
    \draw[->, red!50, shorten >=2pt, shorten <=2pt] (S2) -- (Rho);
    \draw[->, red!50, shorten >=2pt, shorten <=2pt] (S3) -- (Rho);
    
    \begin{scope}[xshift=-3.3cm, yshift=0cm]
        \node[align=center, font=\bfseries] at (0, 1.6) {Simplex Weights};
        \draw[thick, fill=green!10, draw=green!60!black] (0,1.1) -- (-1.1,-0.4) -- (1.1,-0.4) -- cycle;
        \node[circle, fill=green!60!black, inner sep=2pt, label={[font=\bfseries, text=green!40!black]right:$\mathbf{p}$}] at (0.25, 0.2) {};
        
        \node[align=center, anchor=north] at (0, -0.6) {$\mathbf{p} \sim \text{Dir}_K(D)$ \\[0.3em] $p_1 \approx 0.52$ \\ $p_2 \approx 0.20$ \\ $p_3 \approx 0.28$};
    \end{scope}
\end{scope}

\begin{scope}[xshift=6.8cm]
    
    \draw[bloch faint, fill=gray!4] (0,0) circle (\R);
    \draw[equator faint] (\R,0) arc (0:180:\R cm and 0.45cm);
    \draw[equator faint] (\R,0) arc (0:-180:\R cm and 0.45cm);
    
    \draw[axis line] (0,-2.1) node[below] {$|1\rangle$} -- (0,2.1) node[above] {$|0\rangle$};
    
    \coordinate (U1) at (\Aone:\Yone*\R);
    \coordinate (U2) at (\Atwo:\Ytwo*\R);
    \coordinate (U3) at (\Athree:\Ythree*\R);
    \coordinate (Rho) at (\Rhox, \Rhoy); 
    
    \draw[unnorm vector] (0,0) -- (U1);
    \draw[unnorm vector] (0,0) -- (U2);
    \draw[unnorm vector] (0,0) -- (U3);
    
    \node[unnorm point, label={[orange!90!black]above right:$|\tilde{g}_1\rangle$}] at (U1) {}; 
    \node[unnorm point, label={[orange!90!black]left:$|\tilde{g}_2\rangle$}] at (U2) {}; 
    \node[unnorm point, label={[orange!90!black]below right:$|\tilde{g}_3\rangle$}] at (U3) {}; 
    
    \coordinate (S1) at (\Aone:\R);
    \coordinate (S2) at (\Atwo:\R);
    \coordinate (S3) at (\Athree:\R);
    
    \draw[->, dashed, orange!80!black] (U1) -- (S1);
    \draw[->, dashed, orange!80!black] (U2) -- (S2);
    \draw[->, dashed, orange!80!black] (U3) -- (S3);
    
    \node[pure state, inner sep=1.2pt] at (S1) {};
    \node[pure state, inner sep=1.2pt] at (S2) {};
    \node[pure state, inner sep=1.2pt] at (S3) {};
    
    \node[mixed state, label={[red!80!black]below:$\rho_{\text{induced}}^{(K)}$}] at (Rho) {};
    
    \draw[->, red!50, shorten >=2pt, shorten <=2pt] (S1) -- (Rho);
    \draw[->, red!50, shorten >=2pt, shorten <=2pt] (S2) -- (Rho);
    \draw[->, red!50, shorten >=2pt, shorten <=2pt] (S3) -- (Rho);
    
    \begin{scope}[xshift=3.3cm, yshift=0cm]
        \node[align=center, font=\bfseries] at (0, 1.6) {Ginibre Matrix};
        \node[align=center, fill=orange!10, draw=orange!80, rounded corners, inner sep=5pt] 
        at (0, 0.4) {Columns of $G$:\\[0.2em] $|\tilde{g}_j\rangle = \sum_i G_{ij} |i\rangle$};
        
        \node[align=center, fill=gray!10, draw=gray!50, rounded corners, inner sep=5pt] 
        at (0, -1.2) {Trace Normalization:\\[0.2em] $\rho = \frac{GG^\dagger}{\Tr GG^\dagger }$};
    \end{scope}
\end{scope}

\draw[thick, <->, >=Stealth, decorate, decoration={amplitude=1pt, segment length=10pt}] 
    (2.4, 0) -- node[above, font=\bfseries\large] {$\overset{d}{=}$} (4.4, 0);
\end{tikzpicture}
\caption{Geometric illustration of the equivalence between the Mai-Alquier (MA) construction (left) and the induced-state construction (right) for a qubit system ($D=2$) and environment dimension $K=3$. \textit{Left:} A mixed state is formed by convexly combining $K$ independent Haar-random pure states $|\psi_j\rangle$ with weights $\mathbf p\sim \mathrm{Dir}_K(D)$. \textit{Right:} The columns $|\tilde g_j\rangle$ of a $2\times 3$ Ginibre matrix (plotted as  orange circles) determine normalized directions $|\psi_j\rangle=|\tilde g_j\rangle/\|\ket{\tilde g_j}\|$ (plotted as blue circles) and radial weights through trace normalization of $GG^\dagger$. Under this decomposition, the normalized radial variables generate the same Dirichlet law for the mixture weights, yielding the same random density-matrix distribution.}
\label{fig:equivalence}
\end{figure*}
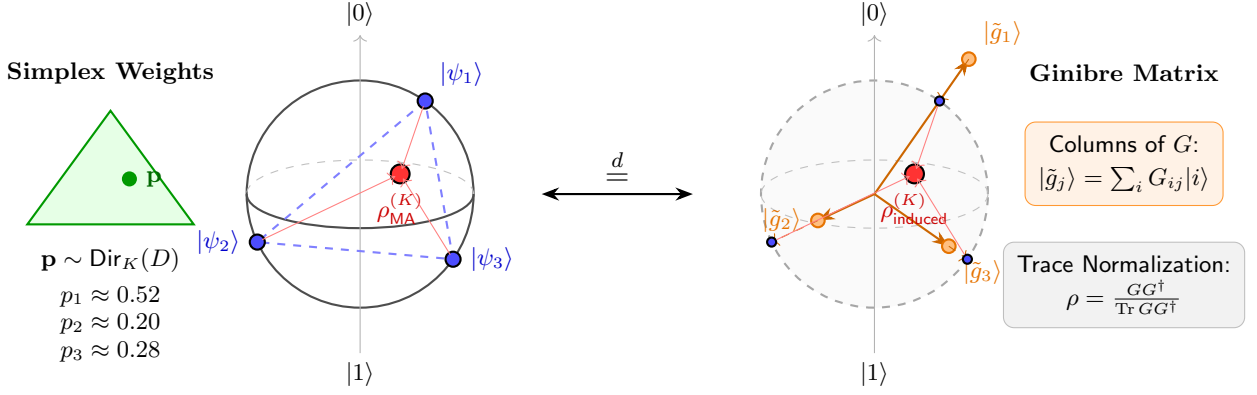

To connect the physical process of environmental coupling with statistical mixtures, we first define the two underlying generative frameworks. Historically, quantum information theory and RMT have generated random density matrices via the physical principle of purification \cite{zyczkowski2001induced}. Let $G$ be a $D \times K$ complex Ginibre matrix \cite{ginibre1965statistical} with independent standard complex Gaussian entries, $G_{ij} \sim \mathcal{CN}(0,1)\equiv\mathcal{N}\left(0,\frac{1}{2}\right)+\imag\mathcal{N}\left(0,\frac{1}{2}\right)$. 
The column-major vectorization of $G$ defines an unnormalized bipartite vector $|\tilde{\Psi}_{AB}\rangle = \mathrm{vec}(G)$ in the composite Hilbert space $\mathcal{H}_A \otimes \mathcal{H}_B$, where $\dim\mathcal{H}_A=D$ and $\dim\mathcal{H}_B=K$. Its normalized counterpart $|\tilde{\Psi}_{AB}\rangle/\|\ket{\tilde{\Psi}_{AB}}\|$ is a uniformly Haar-random pure state on $\mathcal{H}_A \otimes \mathcal{H}_B$.
Operationally, forming the rank-1 operator $|\tilde{\Psi}_{AB}\rangle\langle\tilde{\Psi}_{AB}|$ and tracing out the environment yields the complex Wishart matrix $W = G G^\dagger$. The resulting valid quantum state on the primary system $A$ is then obtained via trace-normalization:
\begin{equation}\label{eq:induced_measure_ref}
\rho_{\text{induced}} = \frac{G G^\dagger}{\Tr G G^\dagger}.
\end{equation}

Alternatively, the MA distribution provides a constructive, operationally intuitive approach based on convex statistical mixing \cite{mai2017pseudo}. It models a $D$-dimensional mixed state explicitly as an ensemble of $K$ independent, single-system Haar-random pure states $|\psi_j\rangle \in \mathcal{H}_A$:
\begin{equation}\label{eq:ma_dist}
\rho_{\text{MA}} = \sum_{j=1}^{K} p_j |\psi_j\rangle\langle\psi_j|,
\end{equation}
where the mixing weights $\mathbf{p} = (p_1, \dots, p_K)$ are drawn from a $K$-dimensional symmetric Dirichlet distribution with concentration parameter $\alpha$, denoted $\mathbf{p} \sim \text{Dir}_{K}(\alpha)$. Here, the parameter $\alpha > 0$ dictates the concentration of the probability vector across the simplex, with larger $\alpha$ favoring more uniform weights ($p_j \approx 1/K$)~(cf. Fig.~1 of Ref.~\cite{lohani2021improving}).

Our central claim is that these two seemingly disparate generative processes---physical partial traces and Dirichlet weighted non-orthogonal statistical mixtures---are identical in distribution for any environmental dimension $K$, provided the Dirichlet concentration parameter matches the system dimension ($\alpha = D$). We formalize this equivalence in the following theorem.

\begin{theorem}[Equivalence of Induced Measures and MA Ensembles]
\label{thm:equivalence}
Let $\rho_{\textup{induced}}$ be a $D$-dimensional random density matrix induced by taking the partial trace over a $K$-dimensional environment of a uniformly random global pure state. Let $\rho_{\textup{MA}}$ be an MA ensemble mixing $K$ independent, single-system Haar-random pure states with classical weights drawn from a symmetric Dirichlet distribution $\mathbf{p} \sim \textup{Dir}_K(D)$. These two ensembles are identical in distribution:
\begin{equation}\label{eq:law_equivalence}
\rho_{\textup{induced}} \overset{\text{d}}{=} \rho_{\textup{MA}}.
\end{equation}
\end{theorem}

\begin{proof}
\textit{Ginibre vectorization and bipartite construction.---}%
Consider a $D \times K$ complex Ginibre matrix $G$ where each entry $G_{ij} \sim \mathcal{CN}(0,1)$. 
Its column-major vectorization defines an unnormalized complex Gaussian vector in $\mathcal{H}_A \otimes \mathcal{H}_B$, which we denote by $|\tilde{\Psi}_{AB}\rangle = \mathrm{vec}(G)$.

By assigning the $j$th column of $G$ to the unnormalized ket $|\tilde{g}_j\rangle = \sum_{i=1}^D G_{ij}\ket{i_A}$, where $\{\ket{i_A}\}_{i=1}^D$ is a fixed orthonormal basis for $\mathcal{H}_A$, we can write the unnormalized bipartite vector as a sum over the environment basis $\{\ket{j_B}\}_{j=1}^K$:
\begin{equation}\label{eq:global_state_vec}
|\tilde{\Psi}_{AB}\rangle = \mathrm{vec}(G)= \sum_{j=1}^K |\tilde{g}_j\rangle \otimes |j_B\rangle.
\end{equation}
Upon normalization, $|\tilde{\Psi}_{AB}\rangle/\|\ket{\tilde{\Psi}_{AB}}\|$ is a Haar-random pure state on $\mathcal{H}_A \otimes \mathcal{H}_B$.

\textit{Partial trace and the Wishart correspondence.---}%
Since the environment basis $\{\ket{j_B}\}_{j=1}^K$ is orthonormal, the partial trace over subsystem $B$ yields a $D \times D$ operator on subsystem $A$. All cross-terms vanish under the partial trace, leaving only the sum of the outer products of the column vectors:
\begin{equation}\label{eq:wishart_equiv}
W = \Tr_B|\tilde{\Psi}_{AB}\rangle\langle\tilde{\Psi}_{AB}| = \sum_{j=1}^K |\tilde{g}_j\rangle \langle \tilde{g}_j| = G G^\dagger.
\end{equation}

\textit{Spherical symmetry and polar factorization.---}%
The entries of the Ginibre matrix $G_{ij}$ are independent standard complex Gaussians, $G_{ij} \sim \mathcal{CN}(0,1)$. Therefore, the joint probability density of any column vector $|\tilde{g}_j\rangle \in \mathbb{C}^D$ depends only on its inner product $\braket{\tilde{g}_j|\tilde{g}_j}$:
\begin{equation}\label{eq:joint_gaussian_pdf}
f(|\tilde{g}_j\rangle) = \frac{1}{\pi^D} e^{-\braket{\tilde{g}_j|\tilde{g}_j}},
\end{equation}
using the fact that the real and imaginary components are sampled independently and each have variance $\sigma^2=\frac{1}{2}$. 
This density is invariant under unitary rotation since $\braket{\tilde{g}_j|U^{\dagger}U|\tilde{g}_j}=\braket{\tilde{g}_j|\tilde{g}_j}$ in the exponent.
In the polar decomposition,
\begin{equation}\label{eq:polar_decomp}
|\tilde{g}_j\rangle = \sqrt{Y_j}|\psi_j\rangle \quad ; \qquad Y_j = \braket{\tilde{g}_j|\tilde{g}_j},
\end{equation}
the directional unit vector $|\psi_j\rangle$ is uniformly distributed on the complex unit hypersphere (defining a Haar-random pure state) \cite{collins2016random}. Furthermore, expressing the isotropic density in polar coordinates factors the distribution into radial and angular parts, ensuring that the directional state $|\psi_j\rangle$ and the squared magnitude $Y_j$ are statistically independent.

Since each complex coordinate $G_{ij}$ is composed of independent real and imaginary parts with variance $1/2$, mapping these coordinates to real variables reveals that the scaled magnitude $Q = 2\braket{\tilde{g}_j|\tilde{g}_j} = 2Y_j$ is the sum of $2D$ independent squared standard real Gaussians. 
Therefore, $Q$ follows a chi-squared distribution with $Q \sim \chi^2(2D)$. Dividing out the scaling factor of $2$, the squared scalar magnitude $Y_j = Q/2$ follows a gamma distribution with shape parameter $D$ and scale parameter $1$ \cite{rice2007mathematical}:
\begin{equation}\label{eq:gamma_magnitude}
Y_j \sim \text{Gamma}(D, 1).
\end{equation}

\textit{Trace normalization and the Dirichlet construction.---}%
Substituting the polar factorization $|\tilde{g}_j\rangle = \sqrt{Y_j}|\psi_j\rangle$ into the trace-normalized induced state separates the mathematical representation into classical scalar weights and quantum projection operators. Noting that the trace of each projector is unity ($\Tr\ket{\psi_j}\bra{\psi_j} = 1$), we expand:
\begin{equation}\label{eq:induced_polar}
\rho_{\text{induced}} = \frac{\sum_{j=1}^K |\tilde{g}_j\rangle \langle \tilde{g}_j|}{\sum_{k=1}^K \braket{\tilde{g}_k|\tilde{g}_k}} = \sum_{j=1}^K \left( \frac{Y_j}{\sum_{k=1}^K Y_k} \right) |\psi_j\rangle\langle\psi_j|.
\end{equation}

Let the classical mixing weights be $p_j = Y_j / \sum_{k=1}^K Y_k$. By standard properties of the gamma and Dirichlet distributions, the sum-normalization of $K$ independent $\mathrm{Gamma}(D,1)$ variables yields a symmetric Dirichlet distribution, so $\mathbf p \sim \mathrm{Dir}_K(D)$.
In addition, because the columns of $G$ are independent and, for each column, the isotropic complex Gaussian law factorizes in polar coordinates into independent radial and angular parts, the joint law of $(Y_1,\dots,Y_K)$ is independent of that of $(|\psi_1\rangle,\dots,|\psi_K\rangle)$. Since $\mathbf p$ is a function only of $(Y_1,\dots,Y_K)$, it follows that the Dirichlet weights are independent of the Haar-random pure states, and hence
\begin{equation}
\rho_{\text{induced}}
= \sum_{j=1}^K p_j |\psi_j\rangle\langle\psi_j|
\end{equation}
has exactly the MA form with concentration parameter $\alpha=D$, completing the proof.
\end{proof}

By establishing that setting the Dirichlet parameter to $\alpha = D$ yields an exact distributional identity, Theorem \ref{thm:equivalence} bridges the physical paradigm of induced measures \cite{zyczkowski2001induced} with classical statistical mixtures. Operationally, this equivalence enables ancilla-free state generation: an experimentalist can reproduce the exact statistics of an induced $D \times K$ mixed state on a single system by classically sampling $\mathbf{p} \sim \text{Dir}_K(D)$ and sequentially preparing independent single-system Haar-random states. 

\section{Analytical Applications of the Framework}
\label{sec:analytical}

The equivalence established in Theorem \ref{thm:equivalence} provides a powerful mathematical framework. By reformulating partial traces as discrete statistical mixtures on classical probability simplices, the MA representation recasts many polynomial expectation values of induced states into combinatorics over Dirichlet moments and low-order Haar averages, avoiding direct integration over composite unitary manifolds. In the following subsections, we deploy this constructive representation to derive exact low-order purity moments, evaluate cross-ensemble distance metrics, and compute exact average determinants.

To streamline the combinatorial evaluations throughout this section and the appendices, let $\Pi_i \equiv |\psi_i\rangle\langle\psi_i|$ denote the rank-1 quantum projector associated with the $i$-th Haar-random pure state. Because distinct states $|\psi_{i_1}\rangle, \dots, |\psi_{i_m}\rangle$ are mutually independent, the expectation value of their cyclic trace factors directly. Using the linearity of the trace and expectation, combined with the Haar 1-design property $\mathbb{E}[\Pi_i] = \frac{\mathbb{I}}{D}$~\cite{ambainis2007quantum, collins2010random}, the trace of any product of $m$ distinct projectors evaluates to
\begin{align}\label{eq:distinct_projector_identity}
    \mathbb{E}\left[\Tr\Pi_{i_1}\Pi_{i_2}\cdots\Pi_{i_m}\right]
    &= \Tr\mathbb{E}[\Pi_{i_1}]\mathbb{E}[\Pi_{i_2}]\cdots\mathbb{E}[\Pi_{i_m}] \nonumber\\
    &= \Tr\left[\left(\frac{\mathbb{I}}{D}\right)^m\right] = \frac{\Tr\mathbb{I}}{D^m} = \frac{1}{D^{m-1}},
\end{align}
for any collection of distinct indices $\{i_1, \dots, i_m\}$ with $m \ge 1$. Together with the idempotency relation $\Pi_i^2 = \Pi_i$, this elementary identity allows arbitrary cyclic expectations to be computed directly by counting distinct state occurrences.

\subsection{Exact low-order purity moments}

To calculate the mean purity, we first algebraically expand the trace of the squared density matrix for the MA ensemble: 
\begin{align}
    \Tr\rho_{\text{MA}}^2 
    &= \Tr\left[ \left(\sum_{i=1}^K p_i \Pi_i\right)^2 \right] = \sum_{i=1}^K\sum_{j=1}^K p_i p_j \Tr\Pi_i \Pi_j.
\end{align}
Because the mixing weights $\mathbf{p}$ and pure-state projectors $\Pi_i$ are statistically independent, the expectation value cleanly separates. For diagonal terms ($i=j$), idempotency gives $\Tr\Pi_i^2 = \Tr\Pi_i = 1$. For off-diagonal terms ($i \neq j$), the quantum trace expectation evaluates directly from Eq.~\eqref{eq:distinct_projector_identity} with $m=2$ as $\mathbb{E}[\Tr\Pi_i\Pi_j] = \frac{1}{D}$.

To evaluate the mixing weights, we use the standard joint moments of the symmetric Dirichlet distribution $\mathbf{p} \sim \text{Dir}_K(D)$:
\begin{align}\label{eq:dirichlet_moments}
    \mathbb{E}\left[\prod_{i=1}^K p_i^{q_i}\right] &= \frac{\Gamma(KD)}{\Gamma\left(KD + \sum_{i=1}^K q_i\right)} \prod_{i=1}^K \frac{\Gamma(D + q_i)}{\Gamma(D)}.
\end{align}
Evaluating this yields $\mathbb{E}[p_i^2] = \frac{D(D+1)}{KD(KD+1)}$ and $\mathbb{E}[p_i p_j] = \frac{D^2}{KD(KD+1)}$ for $i \neq j$. 

Substituting these directly into the separated sum explicitly recovers Lubkin's formula \cite{lubkin1978entropy}:
\begin{align}\label{eq:avg_purity}
    \mathbb{E}\left[\Tr\rho_{\text{MA}}^2\right] 
    &= K \frac{D(D+1)}{KD(KD+1)} \nonumber \\
    &\quad + K(K-1) \frac{D^2}{KD(KD+1)} \frac{1}{D} \nonumber \\
    &= \frac{KD(D+1) + KD(K-1)}{KD(KD+1)} \nonumber \\
    &= \frac{D+K}{KD+1}.
\end{align}

We similarly obtain the exact third moment by expanding $\Tr\rho_{\text{MA}}^3$ into three distinct index partitions:
\begin{enumerate}
    \item \textit{All indices equal ($i=j=k$, $K$ terms):} Idempotency yields $\mathbb{E}[\Tr\Pi_i^3] = 1$. Contributes $K \frac{D(D+1)(D+2)}{KD(KD+1)(KD+2)}$.
    \item \textit{Two indices equal ($i=j \neq k$ and permutations, $3K(K-1)$ terms):} Reducing $\Pi_i^2 = \Pi_i$ and applying Eq.~\eqref{eq:distinct_projector_identity} with $m=2$ gives $\mathbb{E}[\Tr\Pi_i\Pi_k] = \frac{1}{D}$. Contributes $3K(K-1) \frac{D(D+1)}{KD(KD+1)(KD+2)}$.
    \item \textit{All indices distinct ($i, j, k$ pairwise distinct, $K(K-1)(K-2)$ terms):} Applying Eq.~\eqref{eq:distinct_projector_identity} with $m=3$ directly gives $\mathbb{E}[\Tr\Pi_i\Pi_j\Pi_k] = \frac{1}{D^2}$. Contributes $K(K-1)(K-2) \frac{D}{KD(KD+1)(KD+2)}$.
\end{enumerate}
Upon expanding and collecting terms, this simplifies cleanly to Sommers and \.{Z}yczkowski's exact 2004 result \cite{sommers2004statistical}:
\begin{equation}
    \mathbb{E}\left[\Tr\rho_{\text{MA}}^3\right] = \frac{D^2 + K^2 + 3KD + 1}{(KD+1)(KD+2)}.
\end{equation}

The systematic nature of this combinatorial framework allows it to scale methodically to higher-order purity moments. Applying the same partitioning logic for the fourth purity moment (see Appendix~\ref{app:exact_4th_moment} for details), we find
\begin{equation}
\label{eq:purityFourth}
\mathbb{E}\left[\Tr\rho_{\text{MA}}^4\right] = \frac{D^3 + K^3 + 6DK^2 + 6D^2K + 5D + 5K}{(KD+1)(KD+2)(KD+3)},
\end{equation}
which exactly recovers the continuous-integration RMT formula of Sommers and Życzkowski~\cite{sommers2004statistical}. 
We note that this equivalence is consistent with the serendipitous numerical agreement found previously between the purity statistics of MA ensembles and Hilbert-Schmidt-induced states reported in Ref.~\cite{lohani2021improving}, the original motivation for the current work.

\subsection{Mean Hilbert-Schmidt distance between distinct induced ensembles}
\label{sec:meanHS}
The evaluation of distance metrics between quantum states is a key tool in quantum characterization, verification, and tomography. Finding the expected Hilbert-Schmidt distance
\begin{equation}
D_{\text{HS}}^2(\rho_1, \rho_2) = \Tr\!\left[(\rho_1 - \rho_2)^2\right]
\end{equation}
between independent random states drawn from two induced-measure ensembles with environment dimensions $K_1$ and $K_2$ (i.e., $\rho_1 \equiv \rho_{\text{induced}}^{(K_1)}$ and $\rho_2 \equiv \rho_{\text{induced}}^{(K_2)}$) is particularly transparent in the MA representation.

Since $\rho_1$ and $\rho_2$ are generated independently, their joint probability distribution factors as a product measure, $\Pr(\rho_1, \rho_2) = \Pr(\rho_1)\Pr(\rho_2)$. Expanding the squared trace and applying the linearity of trace and expectation yields:
\begin{equation}\label{eq:hs_expansion}
\mathbb{E}\left[D_{\text{HS}}^2(\rho_1, \rho_2)\right] = \mathbb{E}\left[\Tr\rho_1^2\right] + \mathbb{E}\left[\Tr\rho_2^2\right] - 2\mathbb{E}\left[\Tr\rho_1\rho_2\right].
\end{equation}
To evaluate the expected cross-term, we can view this quantity from two mutually illuminating perspectives. Macroscopically, because trace and expectation are linear and the ensembles are statistically independent, the last expectation factors as $\Tr\mathbb{E}[\rho_1]\mathbb{E}[\rho_2]$. For the induced ensembles considered here, unitary invariance implies $\mathbb{E}[\rho_1]=\mathbb{E}[\rho_2]=\frac{\mathbb{I}}{D}$, so this evaluates immediately to
\begin{align}\label{eq:cross_term_macro}
\mathbb{E}\left[\Tr\rho_1\rho_2\right] &= \Tr\mathbb{E}[\rho_1]\mathbb{E}[\rho_2] =  \frac{1}{D}.
\end{align}
Thus the mean cross-term is the dimension-dependent constant $\frac{1}{D}$, independent of the environment dimensions $K_1$ and $K_2$. More generally, the same identity holds for any two independent ensembles satisfying $\mathbb{E}[\rho_1]=\mathbb{E}[\rho_2]=\frac{\mathbb{I}}{D}$; unitary invariance is sufficient, but not necessary, for this mean-state property.

Within the MA representation, the same invariant can also be derived microscopically by resolving each state into its random simplex weights and Haar-random rank-1 projectors. Writing the two states as $\rho_1 = \sum_{i=1}^{K_1} p_{1,i}\Pi_{1,i}$ and $\rho_2 = \sum_{j=1}^{K_2} p_{2,j}\Pi_{2,j}$, where the two ensembles are sampled independently and the simplex weights are independent of the Haar-random projectors within each ensemble, every cross-pair $(\Pi_{1,i},\Pi_{2,j})$ is independent.
Applying Eq.~\eqref{eq:distinct_projector_identity} uniformly across all index pairs and factoring the independent Dirichlet weights and quantum projectors gives:
\begin{align}\label{eq:cross_term_micro}
    \mathbb{E}\left[\Tr\rho_1\rho_2\right] &= \sum_{i=1}^{K_1} \sum_{j=1}^{K_2} \mathbb{E}[p_{1,i}] \mathbb{E}[p_{2,j}] \mathbb{E}\left[\Tr\Pi_{1,i}\Pi_{2,j}\right] \nonumber \\
    &= \frac{1}{D} \left(\sum_{i=1}^{K_1} \mathbb{E}[p_{1,i}]\right) \left(\sum_{j=1}^{K_2} \mathbb{E}[p_{2,j}]\right) = \frac{1}{D}.
\end{align}
Equation~\eqref{eq:cross_term_micro} demonstrates that the MA ensemble fundamentally respects the global centroid geometry. Furthermore, it reveals that the $\frac{1}{D}$ cross-term is a structural consequence of probability conservation on the simplex ($\sum_i p_i = 1$), ensuring that $\mathbb{E}[\Tr\rho_1\rho_2] = \frac{1}{D}$ holds universally for any discrete mixture of Haar-random pure states, including the generalized asymmetric Dirichlet priors examined in Appendix~\ref{app:generalized_dirichlet}.

Substituting the cross-term $\mathbb{E}[\Tr\rho_1\rho_2]=\frac{1}{D}$ together with the purity formula in Eq.~\eqref{eq:avg_purity} into Eq.~\eqref{eq:hs_expansion} yields the exact mean-square Hilbert-Schmidt distance between the two induced ensembles:
\begin{equation}\label{eq:avg_hs_cross_strata}
    \mathbb{E}\left[D_{\text{HS}}^2\left(\rho_1^{(K_1)}, \rho_2^{(K_2)}\right)\right] = \frac{D+K_1}{K_1 D+1} + \frac{D+K_2}{K_2 D+1} - \frac{2}{D}.
\end{equation}
For identical environmental dimensions ($K_1 = K_2 = K$), Eq.~\eqref{eq:avg_hs_cross_strata} simplifies immediately to:
\begin{equation}\label{eq:avg_hs_identical}
\mathbb{E}\left[ D_{\text{HS}}^2\left(\rho_1^{(K)}, \rho_2^{(K)}\right) \right] = \frac{2(D^2-1)}{D(KD+1)},
\end{equation}
recovering the foundational result of Życzkowski and Sommers for induced measures~\cite{zyczkowski2001induced}. 

The general $(K_1,K_2)$ formula in Eq.~\eqref{eq:avg_hs_cross_strata} also appears in Kumar's analysis of Wishart-induced distance statistics~\cite{kumar2020wishart}, where it is obtained through spectral methods in the course of a broader study of distance moments and variances. 
Here, the MA framework provides a unified single-system representation in which both intrastate purities and interstate overlaps reduce to elementary Haar averages and simplex moments. In this derivation, the known results are recovered without direct recourse to continuous group integrals or joint eigenvalue distributions.

\subsection{Average matrix determinants and Gram volumes}
\label{sec:determinant}

For a $D$-dimensional density matrix $\rho$, the determinant is nonzero if and only if $\rho$ is of full rank. Geometrically, for a factorization $A A^\dagger$, $\det AA^\dagger$ represents the squared volume of the $D$-dimensional parallelepiped spanned by the column vectors of $A$.
Within the MA framework, where $\rho_\text{MA} = \sum_{j=1}^K p_j |\psi_j\rangle\langle\psi_j|$, the state is constructed as a convex combination of $K$ rank-1 projectors. Consequently, the rank of $\rho$ is strictly bounded by $K$. If $K < D$, the mixed state is singular and its determinant vanishes identically: $\det\rho_{\text{MA}} = 0$. 

In the regime where $K \ge D$, the state is of full rank with probability 1. Within the quantum physics literature, average determinants of trace-normalized induced measures are traditionally evaluated by integrating over the joint eigenvalue density of the Jacobi unitary ensemble~\cite{sommers2004statistical, zyczkowski2001induced}. Here we show that, within the MA representation, the same quantity admits a direct expansion based on Cauchy-Binet, together with the independence between simplex weights and Haar-random directions.

By representing the density matrix as $\rho = \tilde{\Psi}\tilde{\Psi}^\dagger$, where $\tilde{\Psi} = [\sqrt{p_1}\ket{\psi_1} \cdots \sqrt{p_K}\ket{\psi_K}]$ is the $D \times K$ matrix of weighted pure-state column vectors, we may apply the Cauchy-Binet theorem \cite{horn2012matrix}. For $K \ge D$, this expresses $\det\rho$ as the sum of squared determinants over all $\binom{K}{D}$ possible $D \times D$ column submatrices $\tilde{\Psi}_{\mathbf{j}}$:
\begin{align}\label{eq:det_cauchy_binet}
    \det\rho &= \det\tilde{\Psi}\tilde{\Psi}^\dagger = \sum_{1 \le j_1 < \dots < j_D \le K} \left| \det\tilde{\Psi}_{\mathbf{j}} \right|^2 \nonumber \\
    &= \sum_{1 \le j_1 < \dots < j_D \le K} \left( \prod_{m=1}^D p_{j_m} \right) \left| \det\Psi_{\mathbf{j}} \right|^2,
\end{align}
where $\tilde{\Psi}_{\mathbf{j}} = [\sqrt{p_{j_1}}\ket{\psi_{j_1}} \cdots \sqrt{p_{j_D}}\ket{\psi_{j_D}}]$ is the submatrix formed from the columns indexed by the strictly increasing $D$-tuple $\mathbf{j} = (j_1, \dots, j_D)$, and $\Psi_{\mathbf{j}} = [\ket{\psi_{j_1}} \cdots \ket{\psi_{j_D}}]$ is the corresponding unweighted pure-state matrix. If $K < D$, the sum is empty and $\det\rho = 0$, reflecting rank deficiency.

The classical weights $\mathbf{p}$ are statistically independent of the Haar-random directions $|\psi_j\rangle$; therefore, the expectation value of \cref{eq:det_cauchy_binet} cleanly factors. Under a general Dirichlet distribution, the joint expectation of $D$ distinct coordinates depends on the individual concentration parameters $\alpha_j$. For the symmetric Dirichlet distribution $\mathbf p\sim\mathrm{Dir}_K(\alpha)$, the total concentration is $K\alpha$, so
\begin{align}
\mathbb{E}\left[\prod_{m=1}^D p_{j_m}\right]
&= \frac{\Gamma(K\alpha)}{\Gamma(K\alpha+D)}
\prod_{m=1}^D \frac{\Gamma(\alpha+1)}{\Gamma(\alpha)} \nonumber \\
&= \frac{\Gamma(K\alpha)}{\Gamma(K\alpha+D)}\alpha^D.
\end{align}

The factor
\begin{equation}
|\det\Psi_{\mathbf j}|^2=\det\Psi_{\mathbf j}^\dagger\Psi_{\mathbf j}
\end{equation}
is the Gram determinant of the $D$ independent Haar-random columns of $\Psi_{\mathbf j}$, and hence equals the squared volume of the parallelepiped they span. Let $S_{m-1}=\mathrm{span}\{|\psi_{j_1}\rangle,\dots,|\psi_{j_{m-1}}\rangle\}$, and let $h_m^2$ denote the squared norm of the component of $|\psi_{j_m}\rangle$ orthogonal to $S_{m-1}$, with $h_1^2=1$. Then
\begin{equation}
|\det\Psi_{\mathbf j}|^2=\prod_{m=1}^D h_m^2.
\end{equation}
Because a Haar-random pure state is isotropic, $\mathbb E[|\psi\rangle\langle\psi|]=\frac{\mathbb I}{D}$ and, conditional on $S_{m-1}$, its expected squared projection onto this $(m-1)$-dimensional subspace is $(m-1)/D$. Therefore,
\begin{equation}
\mathbb E[h_m^2\mid S_{m-1}]=1-\frac{m-1}{D}=\frac{D-m+1}{D}.
\end{equation}
Applying iterated expectation gives
\begin{equation}
\mathbb E\!\left[|\det\Psi_{\mathbf j}|^2\right]
=\prod_{m=1}^D \frac{D-m+1}{D}
=\frac{D!}{D^D}.
\end{equation}

We can now evaluate the Cauchy-Binet expansion by summing over all $\binom{K}{D}$ possible column selections. The expected determinant for a MA ensemble with arbitrary concentration parameter $\alpha$ evaluates to:
\begin{equation}
\mathbb{E}\left[\det\rho_{\text{MA}}\right] = \binom{K}{D} \frac{\Gamma(K\alpha) \alpha^D}{\Gamma(K\alpha + D)} \frac{D!}{D^D}.
\end{equation}
This equation provides the exact mean determinant for any phenomenological MA prior. To bridge this back to physical open quantum systems, we recall Theorem \ref{thm:equivalence}, which establishes that the MA ensemble precisely mirrors standard induced RMT measures when the concentration parameter is tuned to the system dimension ($\alpha = D$). Substituting this physical constraint yields:
\begin{align}\label{eq:det_symmetric_recover}
\mathbb{E}\left[\det\rho_{\text{induced}}\right] &= \frac{K!}{D!(K-D)!} \frac{\Gamma(KD) D^D}{\Gamma(KD + D)} \frac{D!}{D^D}  \nonumber \\
&= \frac{K!}{(K-D)!} \frac{\Gamma(KD)}{\Gamma(KD + D)} \nonumber \\
&= \frac{\prod_{j=1}^{D}(K-j+1)}{\prod_{j=1}^{D}(KD+j-1)}.
\end{align}

To confirm the physical and mathematical consistency of Eq.~\eqref{eq:det_symmetric_recover}, we observe that it seamlessly reproduces two special cases. First, when $K < D$, the numerator product $\prod_{j=1}^{D}(K-j+1)$ contains the vanishing factor $(K-K)=0$ at $j=K+1$, ensuring that $\mathbb{E}[\det\rho_{\text{induced}}] = 0$ as required for rank-deficient, singular density matrices. Second, evaluating Eq.~\eqref{eq:det_symmetric_recover} for a Hilbert-Schmidt ensemble of qubits ($D=2, K=2$) yields $\mathbb{E}[\det\rho] = \frac{2 \times 1}{4 \times 5} = \frac{1}{10}$. This perfectly matches the exact geometric integration of the determinant, $\det\rho = \frac{1}{4} - r^2$, over the normalized 3-ball Bloch vector radial distribution $f_{\text{HS}}(r) = 24r^2$ derived by \.{Z}yczkowski and Sommers~\cite{zyczkowski2001induced} (adopting their convention where the Bloch radius satisfies $r \in [0, 1/2]$), which yields $\int_0^{\frac{1}{2}} (\frac{1}{4} - r^2) 24r^2 \, dr = \frac{1}{10}$.

\section{Conclusion and Outlook}
\label{sec:conclusion}

In this work, we established the mathematical equivalence between the family of induced measures for random mixed quantum states and the MA distribution. By formalizing this relationship, we mapped the physical process of environmental tracing to a multivariate statistical mixture of independent, non-orthogonal Haar-random pure states weighted by a Dirichlet simplex. 
This equivalence provides a powerful, unifying computational alternative to traditional calculation approaches. By translating high-dimensional continuous integrals into discrete combinatorics, we seamlessly recovered exact low-order purity moments, exact cross-ensemble Hilbert-Schmidt distance metrics, and exact average matrix determinants.

Looking forward, this combinatorial and convex geometric approach opens new computational pathways for analyzing open quantum systems. Because the MA framework cleanly separates the quantum geometry of pure states (the Haar measure) from classical statistical mixing (the Dirichlet measure), computationally rigid problems can be reformulated into low-dimensional convex geometry. In addition to its potential computational advantages for problems involving the Hilbert-Schmidt ensemble, our result may have implications for quantum foundations.

Recent work has shown that proper and improper realizations of the same reduced density operator---although predictively indistinguishable---can lead to different retrodictive updates when used as priors for so-called ``prior-extended'' Petz maps \cite{LiuPRL2026Proper}. The equivalence established here therefore provides a natural setting in which to compare the retrodictive behaviour of proper-mixture beliefs generated through the MA construction with that of improper-mixture beliefs arising from the induced construction. An interesting direction for future work would be to determine whether these constructions induce distinct probability distributions over retrodiction maps, characterize their dependence on $D$ and $K$, and quantify any resulting differences in quantum-state recovery or other noisy inverse problems.

\begin{acknowledgments}
J.M.L. acknowledges funding from the National Science Foundation (ECCS-2540189) and the U.S. Department of Energy (DE-SC0026396). S.L. acknowledges support from the Southern Methodist University (SMU) startup fund and the Sam Taylor Fellowship Program.
\end{acknowledgments}

\appendix

\section{Combinatorial derivation of fourth moment}
\label{app:exact_4th_moment}

Let $\Pi_i = |\psi_i\rangle\langle\psi_i|$ denote the rank-1 quantum projection operator for the $i$-th pure state in the MA ensemble. Expanding $\Tr\rho_{\text{MA}}^4$ over four sum indices yields:
\begin{equation}
    \mathbb{E}\left[\Tr\rho_{\text{MA}}^4\right] = \sum_{i,j,k,l=1}^K \mathbb{E}\left[ p_i p_j p_k p_l \right] \mathbb{E}\left[ \Tr\Pi_i \Pi_j \Pi_k \Pi_l\right].
\end{equation}
Since every term in the fourth-moment expansion is quartic in the Dirichlet weights, \(\sum_{r=1}^K q_r = 4\) in Eq.~\eqref{eq:dirichlet_moments}, and hence all cases share the common denominator $\left[\prod_{m=0}^3 (KD+m)\right]^{-1}$. The four-index expansion decomposes into five integer partitions mapping to six distinct trace topologies as summarized in Table~\ref{tab:4th_moment_cases}. 
Note that Case 4 (crossing topology $\Tr\Pi_i \Pi_j \Pi_i \Pi_j = |\langle\psi_i|\psi_j\rangle|^4$) requires the second raw moment of the $\text{Beta}(1, D-1)$ overlap distribution, $\mathbb{E}[|\langle\psi_i|\psi_j\rangle|^4] = \frac{2}{D(D+1)}$.

\begin{table*}[t]
\caption{Combinatorial partitions, Dirichlet weight numerators, and expected quantum trace topologies for the fourth purity moment $\mathbb{E}\left[\Tr\rho_{\text{MA}}^4\right]$.}
\label{tab:4th_moment_cases}
\begin{ruledtabular}
\renewcommand{\arraystretch}{1.2}
\begin{tabular}{c l c c c}
\textbf{Case} & \textbf{Index partition} & \textbf{Permutations} & \textbf{Dirichlet weight factor} & \textbf{Trace expectation} \\ \hline
1 & $i=j=k=l$ ($q_i=4$) & $K$ & $D(D+1)(D+2)(D+3)$ & $1$ \\
2 & $q_i=3, q_j=1$ & $4K(K-1)$ & $D^2(D+1)(D+2)$ & $\frac{1}{D}$ \\
3 & $q_i=2, q_j=2$ (adjacent) & $2K(K-1)$ & $D^2(D+1)^2$ & $\frac{1}{D}$ \\
4 & $q_i=2, q_j=2$ (crossing) & $K(K-1)$ & $D^2(D+1)^2$ & $\frac{2}{D(D+1)}$ \\
5 & $q_i=2, q_j=1, q_k=1$ & $6K(K-1)(K-2)$ & $D^3(D+1)$ & $\frac{1}{D^2}$ \\
6 & $q_i=1, q_j=1, q_k=1, q_l=1$ & $K(K-1)(K-2)(K-3)$ & $D^4$ & $\frac{1}{D^3}$ \\
\end{tabular}
\end{ruledtabular}
\end{table*}

Summing the product of permutations, weight numerators, and quantum traces from Table~\ref{tab:4th_moment_cases} across the shared Dirichlet denominator yields:
\begin{widetext}
\begin{align}
    \mathbb{E}\left[\Tr\rho_{\text{MA}}^4\right] 
    &= \frac{1}{KD(KD+1)(KD+2)(KD+3)} \Bigg[ \overbrace{K \cdot D(D+1)(D+2)(D+3) \cdot 1}^{\text{Case 1}} \nonumber \\
    &\quad + \overbrace{4K(K-1) \cdot D^2(D+1)(D+2) \cdot \frac{1}{D}}^{\text{Case 2}} + \overbrace{2K(K-1) \cdot D^2(D+1)^2 \cdot \frac{1}{D}}^{\text{Case 3}} \nonumber \\
    &\quad + \overbrace{K(K-1) \cdot D^2(D+1)^2 \cdot \frac{2}{D(D+1)}}^{\text{Case 4}} + \overbrace{6K(K-1)(K-2) \cdot D^3(D+1) \cdot \frac{1}{D^2}}^{\text{Case 5}} \nonumber \\
    &\quad + \overbrace{K(K-1)(K-2)(K-3) \cdot D^4 \cdot \frac{1}{D^3}}^{\text{Case 6}} \Bigg].
\end{align}

Factoring out $KD$ from the numerator cancels the leading $KD$ term in the denominator:
\begin{align}
    \frac{\text{Numerator}}{KD} 
    &= (D+1)(D+2)(D+3) + 4(K-1)(D+1)(D+2) + 2(K-1)(D+1)^2 \nonumber \\
    &\quad + 2(K-1)(D+1) + 6(K-1)(K-2)(D+1) + (K-1)(K-2)(K-3) \nonumber \\
    &= D^3 + K^3 + 6DK^2 + 6D^2K + 5D + 5K,
\end{align}
which simplifies cleanly to the exact closed-form result in \cref{eq:purityFourth} of the main text.
\end{widetext}

\section{Mathematical identities of generalized Dirichlet-weighted ensembles}
\label{app:generalized_dirichlet}
For purely mathematical interest, one can consider generalized MA ensembles where weights are drawn from an arbitrary asymmetric Dirichlet distribution $\mathbf{p} \sim \text{Dir}(\alpha_1, \dots, \alpha_K)$ with total mass $\alpha_0 = \sum_{k=1}^K \alpha_k$. We emphasize that this construction serves as an \emph{ad hoc} prior rather than a physical partial trace of an open quantum system.

Evaluating the moments over the asymmetric simplex yields the exact expectation for the purity of such ensembles:
\begin{equation}
\mathbb{E}\left[\Tr\rho_{\text{asym}}^2\right] = \frac{(D-1)\sum_{i=1}^K \alpha_i^2 + D\alpha_0 + \alpha_0^2}{D \alpha_0 (\alpha_0 + 1)}.
\end{equation}
Applying the Cauchy-Binet theorem over asymmetric Dirichlet expectations yields the exact average matrix determinant:
\begin{equation}
\begin{split}
    &\mathbb{E}\left[\det\rho_{\text{asym}}\right] =\frac{\Gamma(\alpha_0)}{\Gamma(\alpha_0 + D)} \frac{(D-1)!}{D^{D-1}}  \sum_{\substack{1 \le j_1 < \cdots \\ \cdots < j_D \le K}} \prod_{m=1}^D \alpha_{j_m}.
\end{split}
\end{equation}
When all shape parameters are set equal ($\alpha_i = D$), these equations reduce to the physical symmetric RMT expressions derived in the main text.

\bibliography{references}

\end{document}